\documentclass[letterpaper, 10pt, conference]{ieeeconf}

\IEEEoverridecommandlockouts
\usepackage{amsmath,amssymb,mathtools}
\usepackage{mathrsfs}

\usepackage{graphicx}
\graphicspath{{media/}}
\usepackage{booktabs}
\usepackage[caption=false,font=footnotesize]{subfig}

\usepackage{algorithm}
\usepackage{algpseudocode}

\usepackage{xcolor}      
\usepackage{microtype}
\usepackage{nicefrac}
\usepackage{tikz}
\allowdisplaybreaks

\usepackage[hidelinks]{hyperref}
\usepackage[capitalize,noabbrev]{cleveref}

\newtheorem{theorem}{Theorem}[section]
\newtheorem{proposition}[theorem]{Proposition}
\newtheorem{lemma}[theorem]{Lemma}
\newtheorem{corollary}[theorem]{Corollary}

\newtheorem{definition}[theorem]{Definition}

\newcommand{\E}{\mathbb{E}}

\newcommand{\R}{\mathbb{R}}

\newcommand{\sS}{\mathcal{S}}
\newcommand{\sA}{\mathcal{A}}
\newcommand{\sZ}{\mathcal{Z}}
\newcommand{\sT}{\mathcal{T}}

\newcommand{\diff}{\text{d}}

\title{\LARGE \bf
Convergent Q-Learning for Infinite-Horizon General-Sum Markov Games through Behavioral Economics
}

\author{Yizhou Zhang and Eric Mazumdar
\thanks{Department of Computing and Mathematical Sciences, California Institute of Technology}
}

\begin{document}

\maketitle
\thispagestyle{empty}
\pagestyle{empty}

\begin{abstract}

Risk-aversion and bounded rationality are two key characteristics of human decision-making. Risk-averse quantal-response equilibrium (RQE) is a solution concept that incorporates these features, providing a more realistic depiction of human decision making in various strategic environments compared to a Nash equilibrium. Furthermore a class of RQE has recently been shown in~\cite{mazumdar2024tractableequilibriumcomputationmarkov} to be universally computationally tractable in all finite-horizon Markov games, allowing for the development of multi-agent reinforcement learning algorithms with convergence guarantees.  In this paper, we expand upon the study of RQE and analyze their computation in both two-player normal form games and discounted \emph{infinite}-horizon Markov games. For normal form games we adopt a monotonicity-based approach allowing us to generalize previous results. We first show uniqueness and Lipschitz continuity of RQE with respect to player's payoff matrices under monotonicity assumptions, and then provide conditions on the players' degrees of risk aversion and bounded rationality that ensure monotonicity. We then focus on discounted infinite-horizon Markov games. We define the risk-averse quantal-response Bellman operator and prove its contraction under further conditions on the players' risk-aversion, bounded rationality, and temporal discounting. This yields a Q-learning based algorithm with convergence guarantees for all infinite-horizon general-sum Markov games.

\end{abstract}

\section{INTRODUCTION}
Strategic interactions between machine learning algorithms and other agents---each pursuing its own objective in uncertain, complex environments---arise in numerous real-world applications, including autonomous driving, smart grid and financial trading. For humans confronted with such scenarios, two key features typically characterize their decision-making: risk-aversion and bounded rationality. Risk aversion helps agents be robust against the uncertainty of the environment and the varying behavior of other agents, while bounded rationality reflects their limited information processing ability in such complex contexts.

When these problems of strategic decision making are cast as \emph{games}, numerous works highlight how the classical notion of a Nash equilibrium fails to capture human decision-making (see e.g., ~\cite{Experiment3,mckelvey1995quantal,goeree2003risk,prediction1}). One view is that this is due to the fact that Nash captures neither risk-aversion---since it results from agents maximizing their expected payoff---nor bounded rationality---since agents are assumed to be perfectly rational. Moreover, finding Nash equilibria can be intractable from both a computational perspective \cite{daskalakis2013complexity} or a learning perspective \cite{mertikopoulos2017cyclesadversarialregularizedlearning, mazumdar2020gradient}. 

Motivated by behavioral economics \cite{mckelvey1995quantal}, a recent work \cite{mazumdar2024tractableequilibriumcomputationmarkov} proposed the notion of a risk-averse quantal-response equilibrium, a new solution concept that incorporates risk-aversion and bounded rationality through regularization in normal form and finite-horizon Markov games. They further proved that under certain assumptions on players' degrees of risk aversion and bounded rationality, a class of RQEs can be efficiently computed regardless of the payoff structure of the game, implying that RQEs are also universally computationally tractable across finite horizon Markov games. 

In this paper we generalize the study of the properties of RQEs in two-player normal form games from a monotonicity \cite{rosen1965existence} point of view, and further extend it to discounted infinite-horizon Markov games. In particular, for normal form games we show that under strict monotonicity the RQE is unique, and furthermore that under strong monotonicity, the equilibrium strategies are Lipschitz continuous with respect to the payoff matrices of the game. We then provide the conditions for the game to be monotone in terms of the second-order gradients of the regularizers, and similar to that shown in \cite{mazumdar2024tractableequilibriumcomputationmarkov}, these conditions do not depend on the payoff matrices of the game. This also suggests that RQEs are tractable using existing algorithms designed for computing Nash equilibrium in monotone games \cite{cai2023doubly, golowich2020tight}.

Based on our results in the normal form game, we extend the concept of RQEs to discounted infinite-horizon Markov games by viewing the $Q$ functions in the stage game as payoff matrices a normal form game. We define the risk-averse quantal-response Bellman operator whose fixed point is an RQE of the Markov game, and prove the contraction of this operator under the same condition for strong monotonicity as in normal form games, as long as the discount factor of the Markov game does not exceed certain value that depends on the level of strong monotonicity. Combining contraction with the fact that strong monotonicity holds regardless of $Q$ functions, we provide a $Q$-learning based algorithm using the risk-averse quantal-response Bellman operator that provably converges to an RQE for all discounted two player infinite-horizon general-sum Markov games. This is in sharp contrast to the intractability of even CCEs in discounted general-sum Markov games \cite{daskalakis2022complexitymarkovequilibriumstochastic}.

\section{PRELIMINARIES}
We consider a two-player general-sum bimatrix game with finite pure strategies, where the pure strategies of each player $i\in\{1, 2\}$ is chosen from a set $\sA_i$ indexed by $\{1,2,\dots,|\sA_i|\}$ and the utility of player $i$ for actions $(a_1,a_2)\in \sA_1\times \sA_2$ is given by $R_i(a_1,a_2)$. For a pair of mixed strategies (policies) $\pi=(\pi_1,\pi_2)\in \Delta_{|\sA_1|}\times \Delta_{|\sA_2|}$, the expected utility $U_i$ of player $i$ is given by:
\begin{equation}
    U_i(\pi_1,\pi_2)=\E_{(a_1,a_2)\sim(\pi_1,\pi_2)}[R_i(a_1,a_2)].
\end{equation}
Utilizing the linearity of expectation, we slightly overload the notation $R_i\in \R^{|\sA_i|\times |\sA_{-i}|}$, where $-i$ denotes the player other than player $i$, to represent the payoff matrix defined by $[R_i]_{mn}=R_i(a_i=m,a_{-i}=n)$ and now the utility functions have a simple bilinear form:
\begin{equation}\label{eq:bilinear_utility_function}
    U_i(\pi_i,\pi_{-i};R_i)=\pi_i^T R_i\pi_{-i}; \quad i=1, 2.
\end{equation}

\subsection{Risk-Aversion and Bounded Rationality}
While the bilinear utility functions \eqref{eq:bilinear_utility_function} is an effective and well-studied objective in games, they don't precisely reflect people's behavior in real-world decision making, since people are risk-averse against uncertain outcome introduced by mixed strategies, and are not perfectly rational when making decisions.
Recent work \cite{mazumdar2024tractableequilibriumcomputationmarkov} took risk-aversion and bounded rationality into consideration under the setting of matrix games, and introduced a new objective function for the players. They showed that when the players are using convex risk measures \cite{follmer2002convex} and their policies are constrained to quantal responses, the objective function for player $i$ becomes:
\begin{align}
    &f_i(\pi_i,\pi_{-i};R_i)\nonumber\\
    =&\sup_{p_i\in\Delta_{|\sA_{-i}|}} -\pi_i^TR_ip_i
    -D_i(p_i,\pi_{-i})+\epsilon_i\nu_i(\pi_i)\label{eq:objective_function} ,
\end{align}
where each player seeks to minimize its objective function $f_i(\pi_i,\pi_{-i})$. In \eqref{eq:objective_function}, the first term corresponds to the original bilinear form of player $i$'s expected utility, but taken to be the worst case as if the opponent is using some adversarial policy $p_i$ close to $\pi_{-i}$. The second term $D_i(p_i,\pi_{-i})$ corresponds to the  risk-aversion against the uncertainty of the opponent's action sampled from a stochastic $\pi_{-i}$, as established by the dual representation theorem proposed in \cite{follmer2002convex}. The dual representation theorem also suggests that $D_i(p_i,\pi_{-i})$ can be taken to be convex in $p_i$ for a fixed $\pi_{-i}$. Intuitively, the inner supremum over $p_i$ indicates that player $i$ is trying to be robust to the possible deviation to $p_i$ of player $-i$ from its policy $\pi_{-i}$ regularized by some penalty function $D_i$ representing a notion of distance. The last term $\epsilon_i \nu_i(\pi_i)$, for strictly convex $\nu_i$ constrains the player's strategy to quantal responses \cite{follmer2002convex, mertikopoulos2016learning, sokota2023unifiedapproachreinforcementlearning}, reflecting the bounded rationality of each player.

The solution concept of the game defined by \eqref{eq:objective_function} is the risk-adjusted quantal response equilibrium (RQE), formally defined as follows:
\begin{definition}[\cite{mazumdar2024tractableequilibriumcomputationmarkov}, Definition 5]\label{def:RQE_normal_form}
    A risk-adjusted quantal response equilibrium (RQE) of a two-player general-sum bimatrix game whose payoff matrix is given by $\mathbf{R}=(R_1,R_2)$ is a pair of mixed strategies $\pi^*=(\pi_1^*,\pi_2^*)\in \Delta_{|\sA_1|}\times \Delta_{|\sA_2|}$ such that
    \begin{equation}
        f_i(\pi_i^*,\pi_{-i}^*;\mathbf{R})\leq f_i(\pi_i,\pi_{-i}^*;\mathbf{R}), \forall \pi_i\in\Delta_{|\sA_i|}
    \end{equation}
    for both $i\in\{1, 2\}$. When the RQE is unique, we use $\texttt{RQE}_i$ to denote the value of player $i$ at this equilibrium:
    \begin{equation}
        \texttt{RQE}_i(\mathbf{R}):=f_i(\pi_i^*,\pi_{-i}^*;\mathbf{R}),
    \end{equation}
    again for both $i\in\{1,2\}$.
\end{definition}

In order to process the supremum over $p_i$ in equation \eqref{eq:objective_function}, we introduce two additional adversaries controlling $p_i$ for $i=1,2$. Now two original players and two introduced adversary players form a 4-player game, where each original player $i$ controls $\pi_i$ to minimize the following objective function:
\begin{equation}\label{eq:4playergame_original_player_objective}
    J_i(\pi_i,\pi_{-i},p;\mathbf{R})=-\pi_i^TR_ip_i-D_i(p_i,\pi_{-i})+\epsilon_i\nu_i(\pi_i),
\end{equation}
and each adversary $i$ controls $p_i$ to minimize the following:
\begin{equation}\label{eq:4playergame_adversary_objective}
    \bar{J}_i(\pi,p_i,p_{-i};\mathbf{R})=\pi_i^TR_ip_i+D_i(p_i,\pi_{-i})-\epsilon_i\nu_i(\pi_i)
\end{equation}
we use $z=(\pi,p)$ to denote the joint strategy of the 4-player game.
\cite{mazumdar2024tractableequilibriumcomputationmarkov} showed that each Nash equilibrium $z^*=(\pi^*,p^*)$ in the 4-player game has a corresponding RQE $\pi^*$ in the original two-player game. Therefore when the RQE is unique, we further have $\texttt{RQE}_i(\mathbf{R})=J_i(\pi^*,p^*;\mathbf{R})$, which links the RQE value to the modified 4-player game.

\subsection{Discounted Infinite-Horizon Markov Games}
A discounted two-player infinite-horizon general-sum Markov game is specified by a tuple $\mathcal{MG}=\{\sS, \{\sA_i\}_{i=1,2},\{r_i\}_{i=1,2}, \gamma, P \}$ where $\sS$ is the state space of the underlying MDP, $\sA_i$ is the action space of player $i\in\{1, 2\}$, and we use the notation $\sA=\sA_1\times \sA_2$ to denote the product action space of both players. We assume both $|\sS|$ and $|\sA_1|, |\sA_2|$ to be finite. $r_i:\sS\times \sA\rightarrow [0, 1]$ is the reward function of player $i$, which we assume to be deterministic. We use $\mathbf{r}$ to denote the paired reward function $\mathbf{r}:=(r_1,r_2)$. $\gamma\in[0,1)$ is the discount factor and $P:\sS\times \sA\rightarrow \Delta_\sS$ is the transition kernel, where $P(s'|s,\mathbf{a})$ is the probability of the next state being $s'$ given the current state $s$ and the current actions $\mathbf{a}=(a_1,a_2)$ of the players.

We focus on the Markov policies, the class of policies where the action selection probability only depends on the current state instead of the entire gameplay trajectory, i.e. $\pi=(\pi_1,\pi_2)$ where $\pi_i:\sS\rightarrow \Delta_{|\sA_i|}, i\in\{1,2\}$.
Given a product Markov policy $\pi$, without considering risk-aversion and bounded rationality, player $i$ has an expected discounted cumulative reward given by $\E_\pi[\sum_{t=0}^\infty \gamma^t r_i(s_t,a_t)]$.

To incorporate risk-aversion in discounted infinite-horizon Markov games, we consider two sources of randomness, one against the uncertainty of the other player's strategy, and the other against the uncertainty in the state transition. Given a product policy $\pi$, we recursively define the value function $V_i^\pi:\sS\rightarrow \R$ and $Q$ function $Q_i^\pi:\sS\times \sA\rightarrow \R$ of player $i$ as:
\begin{align*}
    &V_i^\pi(s)=-\sup_{p_i\in\Delta_{|\sA_{-i}|}}-\pi_i^T(s)Q_i^\pi(s,\cdot)p_i-D_i\left(p_i,\pi_{-i}(s)\right);\\
    &Q_i^\pi(s,\mathbf{a})=r_i(s,\mathbf{a})\\
    &\quad +\gamma \inf_{\widetilde{P}\in \Delta_\sS}\left\{\E_{s'\sim \widetilde{P}}V_i^\pi(s')+D_i^{\text{env}}\left(\widetilde{P},P(\cdot|s,\mathbf{a})\right)\right\}.
\end{align*}
Here in the definition of $Q_i^\pi$ we follow Theorem 2 in \cite{zhang2024softrobustmdpsrisksensitive} to convert risk-aversion against the environment to taking an infimum over $\widetilde{P}\in\Delta_\sS$.

Similar to what we did in normal form games, we assume bounded rationality by constraining the policies to quantal responses by adding a regularization term $\epsilon_i\nu_i(\pi_i)$, then the value function becomes:
\begin{equation}\label{eq:V_epsilon_def}
    V_i^{\epsilon_i,\pi}(s)=-f_i(\pi(\cdot|s);Q_i^{\epsilon_i,\pi}(s,\cdot))
\end{equation}
for $f_i$ defined in \eqref{eq:objective_function}, and the $Q$ function accordingly:
\begin{align}
    &Q_i^{\epsilon_i,\pi}(s,\mathbf{a})=r_i(s,\mathbf{a})\nonumber\\
    &+ \gamma \inf_{\widetilde{P}\in \Delta_\sS}\left\{\E_{s'\sim \widetilde{P}}V_i^{\epsilon_i,\pi}(s')+D_i^{\text{env}}\left(\widetilde{P},P(\cdot|s,\mathbf{a})\right)\right\}.\label{eq:Q_epsilon_def}
\end{align}

Now we can generalize the definition of RQE to Markov games through the defined value functions as follows:
\begin{definition}(Stationary Markov RQE)
    A pair of Markov policies $\pi^*=(\pi_1^*,\pi_2^*)$ is said to be an RQE of a two-player Markov game $\mathcal{MG}$ if for both $i\in\{1,2\}$ and for all states $s\in \sS$,
    \begin{equation}
        V_i^{\epsilon_i,\pi^*}(s)\geq V_i^{\epsilon_i,\pi_i,\pi_{-i}^*}(s), \forall \pi_i: \sS\rightarrow \Delta_{|\sA_i|}.
    \end{equation}
\end{definition}
It is worth noticing that unlike the one-step bounded rationality considered in \cite{mazumdar2024tractableequilibriumcomputationmarkov}, we add the regularization term $\epsilon_i\nu_i(\pi_i)$ at every time step, which is also widely adopted in sequential decision making \cite{Ortega_2013, evans2023boundedrationalityrelaxingbest}.

\section{PROPERTIES OF RQE FROM A VIEW OF MONOTONICITY}
In this section we focus on normal form games. We show the uniqueness and Lipschitz continuity of RQE with respect to its payoff matrix under certain conditions. We first obtain our results through monotonicity of the modified 4-player game and then discuss the conditions for monotonicity to hold. We first give the definition of a monotone game:
\begin{definition}
    Consider a $N$-player game where player $i$ chooses action $z_i$ from a compact and convex action space $\sZ_i$ to minimize an objective function $J_i(z_i,z_{-i})$. Let $\sZ=\prod_{i=1}^N \sZ_i$ denote the product action space, the gradient operator $F:\sZ\rightarrow \R^N$ is defined as $F=(F_1,F_2,\dots, F_N)^T$,
    where $F_i:\sZ_i\rightarrow \R$ is the gradient of $J_i$ with respect to $z_i$ given by $F_i(z)=\nabla_{z_i} J_i(z_i,z_{-i})$. A game is monotone if its gradient operator $F$ is monotone, i.e.
    \begin{align*}
        (z-z')^T\left(F(z)-F(z')\right)\geq 0, \forall z,z'\in \sZ.
    \end{align*}
    The game is strictly monotone when the inequality above is strict,
    and is $\alpha$-strongly monotone if $F$ is $\alpha$-strongly monotone:
    \begin{align*}
        (z-z')^T\left(F(z)-F(z')\right)\geq \alpha\|z-z'\|_2^2, \forall z,z'\in \sZ.
    \end{align*}
\end{definition}

\subsection{Uniqueness and Lipschitz Continuity of RQE}
If we consider the modified 4-player game where $z=(\pi_1,\pi_2,p_1,p_2)$, the objective function is given by:
\begin{align}
    F(z;\mathbf{R})&=\begin{bmatrix}
        \nabla_{\pi_1}J_1(z;\mathbf{R})\\
        \nabla_{\pi_2}J_2(z;\mathbf{R})\\
        \nabla_{p_1}\bar{J}_1(z;\mathbf{R})\\
        \nabla_{p_2}\bar{J}_2(z;\mathbf{R})
    \end{bmatrix}
    =\begin{bmatrix}
        -R_1p_1+\epsilon_1\nabla\nu_1(\pi_1)\\
        -R_2p_2+\epsilon_2\nabla\nu_2(\pi_2)\\
        R_1^T\pi_1+\nabla_{p_1}D_1(p_1,\pi_2)\\
        R_2^T\pi_2+\nabla_{p_2}D_2(p_2,\pi_1)
    \end{bmatrix}.\label{eq:Fz_definition}
\end{align}
Given the 4-player game is strictly monotone, we have the following:
\begin{proposition}\label{prop:normal_form_RQE_uniqueness}
    If the modified 4-player game is strictly monotone, or equivalently $F(z)$ is a strictly monotone operator, the RQE of the original 2-player game is unique.
\end{proposition}
\begin{proof}
    Since the 4-player game is strictly monotone, following the same proof of Theorem 2 in \cite{rosen1965existence} we obtain uniqueness of its Nash equilibrium, and by Proposition 1 in \cite{mazumdar2024tractableequilibriumcomputationmarkov}, each RQE in the original two-player game can be mapped to a different Nash equilibrium in the modified 4-player game. The uniqueness of the RQE in the original 2-player game follows immediately. For completeness we include the proof details as follows: We consider the KKT conditions for the objective functions \eqref{eq:4playergame_original_player_objective} and \eqref{eq:4playergame_adversary_objective} over the simplex of each player given $\mathbf{R}$, which must be satisfied at equilibrium point $z^*$. For the original player we have:
    \begin{equation}\label{eq:KKT_cond_originalplayer}\begin{aligned}
        -R_ip_i^*+\epsilon_i\nabla\nu_i(\pi_i^*)-\lambda(\pi_i)+\mu(\pi_i)1=0;\\
        \pi_i^*\in \Delta_{|\sA_i|}; \lambda(\pi_i)\geq 0; \mu(\pi_i)\in\R;\lambda(\pi_i)^T\pi_i^*=0.
    \end{aligned}\end{equation}
    where $i\in\{1, 2\}$, $\lambda(\pi_i)$ and $\mu(\pi_i)$ are Lagrange multipliers with respect to the simplex constraint, and $\lambda(\pi_i)^T\pi_i^*=0$ denotes complimentary slackness. For the adversaries we have:
    \begin{equation}\label{eq:KKT_cond_adversary}\begin{aligned}
            R_i\pi_i^* +\nabla_pD_i(p_i^*,\pi_{-i}^*)-\lambda(p_i)+\mu(p_i)1=0;\\
            p_i^*\in\Delta_{|\sA_{-i}|}; \lambda(p_i)\geq 0; \mu(p_i)\in \R; \lambda(p_i)^Tp_i^*=0.
    \end{aligned}\end{equation}
    where similarly $i\in\{1,2\}$, $\lambda(p_i), \mu(p_i)$ are Lagrange multipliers and $\lambda(p_i)^Tp_i^*=0$ denotes complimentary slackness. We can combine \eqref{eq:KKT_cond_originalplayer} and \eqref{eq:KKT_cond_adversary} in a more compact form:
    \begin{equation}
        F(z^*; \mathbf{R})=\begin{bmatrix}
        \lambda(\pi_1)-\mu(\pi_1) 1\\
        \lambda(\pi_2)-\mu(\pi_2) 1\\
        \lambda(p_1)-\mu(p_1)1\\
        \lambda(p_2)-\mu(p_2)1
    \end{bmatrix}.
    \end{equation}
    Therefore, for arbitrary $z\in \sZ$ we have:
    \begin{align}
        &(z^*-z)^TF(z^*; \mathbf{R})\nonumber\\
        =&(z^*-z)^T\begin{bmatrix}
        \lambda(\pi_1)-\mu(\pi_1) 1\\
        \lambda(\pi_2)-\mu(\pi_2) 1\\
        \lambda(p_1)-\mu(p_1)1\\
        \lambda(p_2)-\mu(p_2)1
    \end{bmatrix}
    =-z^T\begin{bmatrix}
        \lambda(\pi_1)\\
        \lambda(\pi_2)\\
        \lambda(p_1)\\
        \lambda(p_2)
    \end{bmatrix}\leq 0 \label{eq:optimality_equilibrium_point}
    \end{align}
    where in the second equality we have used the fact that $z,z^*\in \sZ$ such that $(z^*-z)^T1=0$ and complementary slackness, and the last inequality holds because $z\in \sZ$ and $\lambda(\pi_i),\lambda(p_i)\geq 0$.

    Suppose $z_1$ and $z_2$ are two Nash equilibria of the 4-player game with respect to $\mathbf{R}$, we must have $(z_1-z_2)^TF(z_1; \mathbf{R})\leq 0$ and $(z_2-z_1)^TF(z_2;\mathbf{R})\leq 0$. Adding these two inequalities up we have $(z_1-z_2)^T\left(F(z_1;\mathbf{R})-F(z_2;\mathbf{R})\right)\leq 0$, combining strict monotonicity of $F(\cdot;\mathbf{R})$ we have $z_1=z_2$, indicating the uniqueness of the Nash equilibrium.
\end{proof}

Having established the uniqueness of Nash equilibrium of the modified 4-player game and RQE of the original game through \Cref{prop:normal_form_RQE_uniqueness} under monotonicity, we can see that the equilibrium $z$ is a single-valued mapping from $\R^{|\sA_1|\times |\sA_2|}\times \R^{|\sA_2|\times |\sA_1|}$ to $\sZ$. We now further show that the equilibrium point defined as an implicit function of $\mathbf{R}$ is Lipschitz continuous, given that the $F$ is strongly monotone:
\begin{theorem}\label{thm:RQE_Lipschitz_continuity}
    If the modified 4-player game is $\alpha$-strongly monotone, For two different pairs of payoff matrices $\mathbf{R}$ and $\mathbf{R}'$, their corresponding Nash equilibria $z^*=(\pi_1^*,\pi_2^*,p_1^*,p_2^*)$ and $z^\dagger=(\pi_1^\dagger,\pi_2^\dagger,p_1^\dagger,p_2^\dagger)$ satisfies:
    \begin{align*}
        \|z^*-z^\dagger\|_2\leq 2\left(\sqrt{|\sA_1|}+\sqrt{|\sA_2|}\right)\|\mathbf{R}-\mathbf{R}'\|_\infty \Bigg/\alpha,
    \end{align*}
    as a result, the corresponding RQEs in the original game also satisfies:
    \begin{align*}
        \|\pi^*-\pi^\dagger\|_2\leq 2\left(\sqrt{|\sA_1|}+\sqrt{|\sA_2|}\right)\|\mathbf{R}-\mathbf{R}'\|_\infty\Bigg/\alpha.
    \end{align*}
\end{theorem}
\begin{proof}
    We apply \eqref{eq:optimality_equilibrium_point} to $z^*$ and $z^\dagger$ and obtain:
    \begin{equation*}
        \begin{aligned}
            (z^*-z^\dagger)^TF(z^*; \mathbf{R})\leq 0;
            (z^\dagger-z^*)^TF(z^\dagger; \mathbf{R}')\leq 0.
        \end{aligned}
    \end{equation*}
    adding these two inequalities up, we have:
    \begin{align*}
        (z^*-z^\dagger)^T\left(F(z^*; \mathbf{R})-F(z^\dagger; \mathbf{R}')\right)\leq 0.
    \end{align*}
    We split the left hand side into two difference terms and get:
    \begin{align*}
        &(z^*-z^\dagger)^T\left(F(z^*; \mathbf{R})-F(z^*; \mathbf{R}')\right)\\
        &+(z^*-z^\dagger)^T\left(F(z^*; \mathbf{R}')-F(z^\dagger; \mathbf{R}')\right)\leq 0.
    \end{align*}
    For the first term, expanding it element-wise we can see that all nonlinear parts cancel out:
    \begin{align*}
            &(z^*-z^\dagger)^T\left(F(z^*; \mathbf{R})-F(z^*; \mathbf{R}')\right)\\
            =&\begin{bmatrix}
                \pi_1^*-\pi_1^\dagger \\ \pi_2^*-\pi_2^\dagger \\ p_1^*-p_1^\dagger \\ p_2^*-p_2^\dagger
            \end{bmatrix}^T \cdot \begin{bmatrix}
                -(R_1-R_1')p_1^*\\
                -(R_2-R_2')p_2^*\\
                (R_1-R_1')^T\pi_1^*\\
                (R_2-R_2')^T\pi_2^*
            \end{bmatrix}\\
            \geq& -\|z^*-z^\dagger\|_2\left\|\begin{bmatrix}
                -(R_1-R_1')p_1^*\\
                -(R_2-R_2')p_2^*\\
                (R_1-R_1')^T\pi_1^*\\
                (R_2-R_2')^T\pi_2^*
            \end{bmatrix}\right\|_2\\
            \geq & -2\left(\sqrt{|\sA_1|}+\sqrt{|\sA_2|}\right)\|z^*-z^\dagger\|_2\|\mathbf{R}-\mathbf{R}'\|_\infty.
        \end{align*}
    For the second term, strong monotonicity yields:
    \begin{align*}
        (z^*-z^\dagger)^T\left(F(z^*; \mathbf{R}')-F(z^\dagger; \mathbf{R}')\right)\geq \alpha\|z^*-z^\dagger\|_2^2.
    \end{align*}
    Therefore, combining the bounds above, we have:
    \begin{align*}
            &\alpha\|z^*-z^\dagger\|_2^2\leq \\
        &\quad 2\left(\sqrt{|\sA_1|}+\sqrt{|\sA_2|}\right)\|z^*-z^\dagger\|_2\|\mathbf{R}-\mathbf{R}'\|_\infty,
    \end{align*}
    canceling out $\|z^*-z^\dagger\|_2$ and rearranging terms we get:
    \begin{align*}
        \|z^*-z^\dagger\|_2\leq2\left(\sqrt{|\sA_1|}+\sqrt{|\sA_2|}\right)\|\mathbf{R}-\mathbf{R}'\|_\infty\Bigg/\alpha,
    \end{align*}
    which completes the proof.
\end{proof}

\Cref{thm:RQE_Lipschitz_continuity} suggests that the RQE $z^*$ is continuous in the payoff matrices $\mathbf{R}$ of the players. Notice that we can use any $L_p$ norm for the payoff difference term $\|\mathbf{R}-\mathbf{R}'\|_p$ on the right hand side that only differs by a multiplicative constant depending only on $|\sA_1|$ and $|\sA_2|$, we choose $L_\infty$ norm to directly fit the proof in \Cref{sec:tractability_discounted_MG}. Also, as we will show later, the strong monotonicity of the game doesn't depend on $\mathbf{R}$, indicating that the continuity of RQE holds for all payoff structures once the risk-aversion and bounded rationality regularizers are fixed.

\subsection{Monotonicity Conditions for Regularizers }

For all results established so far, the only property we used is $\alpha$-strong monotonicity in the modified 4-player game. We now provide conditions on the original game under which strong monotonicity holds. We first prove a lemma:
\begin{lemma}\label{lem:strong_monotonicity_iff_condition}
    A differentiable mapping $F:\sZ\rightarrow \R^N$ is $\alpha$-strongly monotone if and only if for each $z$,
    \begin{align*}
        \left(\nabla F(z)+\nabla F(z)^T\right)/2\succeq \alpha I;
    \end{align*}
    and is strictly monotone if (but not only if) for each $z$,
    \begin{align*}
        \nabla F(z)+\nabla F(z)^T\succ 0.
    \end{align*}
\end{lemma}
\begin{proof}
    The condition for strict monotonicity can be found in Proposition 12.3 in \cite{rockafellar2009variational}, we now extend it to strong monotonicity.
    For the `only if' direction, suppose $F$ is $\alpha$-strongly monotone, we have:
    \begin{align*}
        (F(z+tw)-F(z))^T((z+tw)-z)\geq \alpha \|(z+tw)-z\|_2^2
    \end{align*}
    for all $z,w\in \sZ$ and $t>0$. Taking $t\rightarrow 0$, we obtain
    \begin{align*}
        w^T \nabla F(z)^T w \geq \alpha\|w\|_2^2.
    \end{align*}
    Similarly, using an alternative form of the definition that $((z+tw)-z)^T(F(z+tw)-F(z))\geq \alpha \|(z+tw)-z\|_2^2$ we have $w^T \nabla F(z) w\geq \alpha \|w\|_2^2$. Averaging two results and rearranging terms, we obtain:
    \begin{align*}
        w^T\left(\frac{\nabla F(z)+\nabla F(z)^T}{2}-\alpha I\right)w\geq 0, \forall w
    \end{align*}
    this implies $(\nabla F(z)+\nabla F(z)^T)/2\succeq \alpha I$.

    For the `if' direction, we have that $\forall w,z\in \sZ$,
    \begin{align*}
        w^T \nabla F(z) w+w^T \nabla F(z)^T w\geq 2\alpha \|w\|_2^2,
    \end{align*}
    since $w^T \nabla F(z)^T w=w^T \nabla F(z) w$, we deduce that $w^T \nabla F(z) w\geq \alpha \|w\|_2^2$. Consider $z,z'\in \sZ$, let $\varphi(t)=(F(tz'+(1-t)z)-F(z))^T(z'-z)$, we have $\varphi(0)=0$, and
    \begin{align*}
        \varphi'(t)=&(z'-z)^T\nabla F\left(tz'+(1-t)z\right)^T(z'-z)\\
        \geq & \alpha \|z'-z\|_2^2
    \end{align*}
    we deduce that
    \begin{align*}
        \left(F(z')-F(z)\right)^T(z'-z)
        =\varphi(0)+\int_{0}^1 \varphi'(t)\diff t
        \geq \alpha\|z'-z\|_2^2,
    \end{align*}
    which completes the proof.
\end{proof}

To get conditions on the regularizers through \Cref{lem:strong_monotonicity_iff_condition}, we use \eqref{eq:Fz_definition} to obtain:
\begin{align}
    &\frac{\nabla F(z;\mathbf{R})+\nabla F(z;\mathbf{R})^T}{2}\nonumber \\
    =&\begin{bmatrix}
        \epsilon_1\nabla^2\nu_1 & 0 & 0 & \nabla^2_{p\pi} D_2^T/2 \\
        0 & \epsilon_2\nabla^2\nu_2 & \nabla^2_{p\pi} D_1^T/2 & 0 \\
        0 & \nabla^2_{p\pi} D_1/2 & \nabla_p^2 D_1 & 0 \\
        \nabla^2_{p\pi} D_2/2 & 0 & 0 & \nabla_p^2 D_2
    \end{bmatrix}\label{eq:F_Jacobian}
\end{align}
where we have used the shorthand notation $\nu_i=\nu_i(\pi_i)$, $D_i=D_i(p_i,\pi_{-i}), i\in\{1,2\}$. By \Cref{lem:strong_monotonicity_iff_condition}, the modified 4-player game is $\alpha$-strongly monotone if and only if $\frac{\nabla F(z;\mathbf{R})+\nabla F(z;\mathbf{R})^T}{2}\succeq \alpha I$. Notice that the right hand side \eqref{eq:F_Jacobian} can be written as a block-diagonal matrix switching rows/columns 2 and 4, we can further simplify the equivalent condition as:
\begin{equation}\label{eq:simplified_strong_monotonicity_condition}
    \begin{aligned}
        \begin{bmatrix}
        \epsilon_1\nabla^2\nu_1 & \nabla^2_{p\pi} D_2^T/2 \\
        \nabla^2_{p\pi} D_2/2 & \nabla_p^2 D_2
    \end{bmatrix}\succeq \alpha I;\\
    \begin{bmatrix}
        \epsilon_2\nabla^2\nu_2 & \nabla^2_{p\pi} D_1^T/2 \\
        \nabla^2_{p\pi} D_1/2 & \nabla_p^2 D_1
    \end{bmatrix}\succeq \alpha I.
    \end{aligned}
\end{equation}

While this equivalent condition is not straightforward to understand, we summarize some necessary or sufficient conditions in the following proposition:
\begin{proposition}
    A necessary condition for the modified 4-player game to be $\alpha$-strongly monotone is $\epsilon_i\nabla^2\nu_i\succeq \alpha I, \nabla_p^2 D_i\succeq \alpha I, i\in\{1, 2\}$. That is, the quantal response functions are $\alpha$-strongly convex for both players, and the risk-aversion penalty function $D_i$ is $\alpha$-strongly convex in $p_i$ for fixed $\pi_{-i}$, $i\in\{1,2\}$. Moreover, if $D_i$ is jointly convex on both arguments for both players and $\epsilon_i\nabla^2\nu_i-\nabla_\pi^2 D_{-i}/2\succeq \alpha I, \nabla_p^2 D_{-i}/2\succeq \alpha I$, the modified 4-player game is $\alpha$-strongly monotone.
\end{proposition}
\begin{proof}
    For the necessary condition part, notice that \eqref{eq:simplified_strong_monotonicity_condition} implies that $\begin{bmatrix}
        \epsilon_i\nabla^2\nu_i-\alpha I & \nabla^2_{p\pi} D_{-i}^T/2 \\
        \nabla^2_{p\pi} D_{-i}/2 & \nabla_p^2 D_{-i}-\alpha I
    \end{bmatrix}\succeq 0$,
    Schur complement suggests $\epsilon_i\nabla^2\nu_i-\alpha I\succeq 0$ and $\nabla_p^2 D_{-i}-\alpha I\succeq 0$.

    For the sufficient condition part, notice that
    \begin{align*}
            &\begin{bmatrix}
        \epsilon_1\nabla^2\nu_1 & \nabla^2_{p\pi} D_2^T/2 \\
        \nabla^2_{p\pi} D_2/2 & \nabla_p^2 D_2
    \end{bmatrix}=\frac{1}{2}\begin{bmatrix}
        \nabla_\pi^2 D_{-i} & \nabla^2_{p\pi} D_{-i}^T \\
        \nabla^2_{p\pi} D_{-i} & \nabla_p^2 D_{-i}
    \end{bmatrix}\\
    &\qquad +\begin{bmatrix}
        \epsilon_i\nabla^2\nu_i-\nabla_\pi^2 D_{-i}/2 & 0 \\
        0 & \nabla_p^2 D_{-i}/2
    \end{bmatrix}
        \end{align*}
    joint convexity suggests the first term to be positive, and therefore if $\begin{bmatrix}
        \epsilon_i\nabla^2\nu_i-\nabla_\pi^2 D_{-i}/2 & 0 \\
        0 & \nabla_p^2 D_{-i}/2
    \end{bmatrix}\succeq \alpha I$,
    $\alpha$-strong monotonicity holds, which simplifies to $\epsilon_i\nabla^2\nu_i-\nabla_\pi^2 D_{-i}/2\succeq \alpha I, \nabla_p^2 D_{-i}/2\succeq \alpha I$.
\end{proof}

Notice that nearly all the penalty functions (as a notion of distance/divergence) used on the simplex satisfies joint convexity in both arguments, including $f$-divergences for arbitrary convex $f$ and $L_p$ norms for $p\geq 1$. Common regularization functions on the simplex are also strongly convex, including negative entropy or log-barrier. This means if $\nu_i$ and $D_i$ functions are selected within this range, we can always select $\epsilon_i$ such that the modified 4-player game is $\alpha$-strongly monotone for some $\alpha>0$.

\section{TRACTABILITY OF RQE IN DISCOUNTED INFINITE-HORIZON MARKOV GAMES}\label{sec:tractability_discounted_MG}
Based on the uniqueness and Lipschitz continuity properties of RQE, we now develop value iteration methods for computing RQE in discounted infinite-horizon Markov games. Throughout this section we assume the choice of $D_i,\epsilon_i$ and $\nu_i$ satisfies the $\alpha$-strong monotonicity conditions \eqref{eq:simplified_strong_monotonicity_condition}. We first define the risk-averse quantal-response Bellman operator as follows:
\begin{definition}
    Given a two-player discounted Markov game $\mathcal{MG}$, risk-aversion penalty functions $D_i(\cdot,\cdot), D_i^{\text{env}}(\cdot,\cdot)$ and regularizers $\epsilon_i\nu_i(\cdot)$ where $i\in\{1,2\}$, the risk-averse quantal-response Bellman operator $\sT$ maps a $Q$ function pair $\mathbf{Q}=(Q_1,Q_2)$ where $Q_i:\sS\times\sA\rightarrow \R$ to another $Q$ function pair $\sT \mathbf{Q}$ in the same function space, defined elementwise as:
    \begin{align*}
        &(\sT \mathbf{Q})_i(s,\mathbf{a})=r_i(s,\mathbf{a})+\\
        &\gamma \inf_{\widetilde{P}\in\Delta_\sS}\left\{-\E_{s'\sim \widetilde{P}}[\texttt{RQE}_i(\mathbf{Q}(s',\cdot))]+D_i^{\text{env}}\left(\widetilde{P},P(\cdot|s,\mathbf{a})\right)\right\}
    \end{align*}
    here we view $\mathbf{Q}(s',\cdot)$ as a pair of payoff matrices in some normal form game with action space $\sA=\sA_1\times \sA_2$.
\end{definition}

Notice that the strong monotonicity implies uniqueness of RQE for any payoff matrix, for arbitrary $Q$ function pairs $\mathbf{Q}$, the term $\texttt{RQE}_i(\mathbf{Q}(s',\cdot))$ is always well-defined. We first show that under certain conditions $\sT$ is a contraction mapping:
\begin{theorem}
    Under the assumption that \eqref{eq:simplified_strong_monotonicity_condition} holds, if $D_i$ are $L$-Lipschitz metrics that satisfy triangle inequality, and $\gamma$ satisfies $\gamma\leq {\alpha}\Bigg/{\left(\alpha+2L\left(\sqrt{|\sA_1|}+\sqrt{|\sA_2|}\right)\right)}$,
    we have that $\sT$ is a contraction mapping.
\end{theorem}
\begin{proof}
    For two $Q$ function pairs $\mathbf{Q}$ and $\mathbf{Q}'$, we have:
    \begin{align}
        &(\sT \mathbf{Q}-\sT \mathbf{Q}')_i(s,\mathbf{a}) \nonumber \\
        \leq & \gamma \sup_{\widetilde{P}\in\Delta_\sS}\E_{s'\sim \widetilde{P}}\left[\texttt{RQE}_i(\mathbf{Q}'(s',\cdot))-\texttt{RQE}_i(\mathbf{Q}(s',\cdot))\right]\label{eq:proof_contraction_1}
    \end{align}
    we have the following bound for the RQE difference term:
    \begin{align*}
        &\texttt{RQE}_i(\mathbf{R})-\texttt{RQE}_i(\mathbf{R}')
        = -(\pi_i^*)^TR_ip^*_i-D_i(p_i^*,\pi_{-i}^*)\\
        &+\epsilon_i\nu_i(\pi_i^*)+(\pi_i^\dagger)^TR_i'p^\dagger_i+D_i(p_i^\dagger,\pi_{-i}^\dagger)-\epsilon_i\nu_i(\pi_i^\dagger)\\
        \leq & -(\pi_i^\dagger)^TR_ip^*_i-D_i(p_i^*,\pi_{-i}^*)+\epsilon_i\nu_i(\pi_i^\dagger)\\
        &+(\pi_i^\dagger)^TR_i'p^\dagger_i+D_i(p_i^\dagger,\pi_{-i}^\dagger)-\epsilon_i\nu_i(\pi_i^\dagger)\\
        =&(\pi_i^\dagger)^T(R_i'p^\dagger_i-R_ip_i^*)-D_i(p_i^*,\pi_{-i}^*)+D_i(p_i^\dagger,\pi_{-i}^\dagger)\\
        \leq &(\pi_i^\dagger)^T(R_i'-R_i)p_i^* -D_i(p_i^*,\pi_{-i}^*)+D_i(p_i^*,\pi_{-i}^\dagger)
    \end{align*}
    where $(\pi^*,p^*)$ and $(\pi^\dagger,p^\dagger)$ are the RQEs w.r.t. $\mathbf{R}$ and $\mathbf{R}'$ respectively, with the inequalities follow from \Cref{def:RQE_normal_form}. Given that $D_i$ satisfies triangle inequality, we have:
    \begin{equation}\begin{aligned}
        &\texttt{RQE}_i(\mathbf{R})-\texttt{RQE}_i(\mathbf{R}')\\
        \leq& (\pi_i^\dagger)^T(R_i'-R_i)p_i^* + D_i(\pi_{-i}^*, \pi_{-i}^\dagger)\\
        \leq& (\pi_i^\dagger)^T(R_i'-R_i)p_i^* +L\|\pi_{-i}^*-\pi_{-i}^\dagger\|_2\\
        \leq& \left(1+\frac{2L\left(\sqrt{|\sA_1|}+\sqrt{|\sA_2|}\right)}{\alpha}\right)\|\mathbf{R}-\mathbf{R}'\|_\infty.
    \end{aligned}\end{equation}
    Now we take $L_\infty$ norm with respect to all possible $(s,\mathbf{a})$ pairs in \eqref{eq:proof_contraction_1} and obtain:
    \begin{equation}\begin{aligned}
        &\|(\sT \mathbf{Q}-\sT \mathbf{Q}')_i\|_\infty\\
        &\leq \gamma \left(1+2L\left(\sqrt{|\sA_1|}+\sqrt{|\sA_2|}\right)\Bigg/{\alpha}\right) \|\mathbf{Q}-\mathbf{Q}'\|_\infty
    \end{aligned}\end{equation}
    since this upper bound holds for both $i\in\{1,2\}$, the left hand side can be simply rewritten as $\|\sT \mathbf{Q}-\sT \mathbf{Q}'\|_\infty$, therefore $\sT$ is a contraction when
    \begin{equation}
        \gamma \left(1+2L\left(\sqrt{|\sA_1|}+\sqrt{|\sA_2|}\right)\Bigg /\alpha\right)\leq 1.
    \end{equation}
\end{proof}

When $\sT$ is a contraction operator, the Banach fixed point theorem suggests there exists a unique fixed point $\mathbf{Q}^*$ such that $\mathbf{Q}^*=\sT \mathbf{Q}^*$. We now prove that this fixed point corresponds to the an RQE of the Markov game:

\begin{proposition}
    Suppose the risk-averse quantal-response Bellman operator $\sT$ has a unique fixed point $\mathbf{Q}^*$, then the policy $\pi^*$ defined state-wise such that $\pi^*(\cdot|s)$ is the unique RQE of the stage game with payoff matrix pair $\mathbf{Q}^*(s,\cdot)$, then $\pi^*(\cdot|s)$ is an RQE of the Markov game.
\end{proposition}
\begin{proof}
    Since $\pi^*(\cdot|s)$ is the RQE of the stage game, we have that $\forall \pi_i:\sS\rightarrow \Delta_{|\sA_i|}, i\in\{1,2\}$
    \begin{equation}\begin{aligned}
        &\texttt{RQE}_i(\mathbf{Q}^*(s,\cdot))=f_i\left(\pi^*(\cdot|s);\mathbf{Q}^*(s,\cdot)\right)\\
        &\leq f_i\left(\pi_i(s),\pi_{-i}^*(\cdot|s);\mathbf{Q}^*(s,\cdot)\right).
    \end{aligned}\end{equation}
    therefore, combining \eqref{eq:V_epsilon_def} and \eqref{eq:Q_epsilon_def} we have:
    \begin{align*}
        &Q_i^{\epsilon_i,\pi^*}(s,\mathbf{a})=r_i(s,\mathbf{a})+\gamma \inf_{\widetilde{P}\in \Delta_\sS}\Big\{D_i^{\text{env}}\left(\widetilde{P},P(\cdot|s,\mathbf{a})\right)\\
        &-\E_{\widetilde{P}}[f_i(\pi^*(s);\mathbf{Q}^{\epsilon,\pi^*}(s,\cdot))]\Big\}
        =r_i(s,a)+\\
        &\gamma \inf_{\widetilde{P}\in \Delta_\sS}\Big\{D_i^{\text{env}}\left(\widetilde{P},P(\cdot|s,\mathbf{a})\right)-\E_{\widetilde{P}}[\texttt{RQE}_i(\mathbf{Q}(s',\cdot))]\Big\}\\
        =&(\sT \mathbf{Q}^{\epsilon,\pi^*})_i(s,\mathbf{a})
    \end{align*}
    Therefore $\mathbf{Q}^{\epsilon,\pi^*}$ is a fixed point of $\sT$. Since $\sT$ has a unique fixed point, we deduce that $\mathbf{Q}^{\epsilon,\pi^*}=\mathbf{Q}^*$ and further,
    \begin{align*}
        V_i^{\epsilon_i,\pi^*}(s)=-f_i(\pi(s);Q_i^{\epsilon_i,\pi^*}(s,\cdot))=-\texttt{RQE}_i(\mathbf{Q}^*(s,\cdot))
    \end{align*}
    Once this is established, we can use a similar argument to policy improvement theorem to prove that $\pi^*$ is an RQE of the Markov game.
\end{proof}

Now that $\sT$ converges to an RQE of the Markov game, we give a $Q$-learning based iterative algorithm for solving the RQE of the Markov game:
\begin{corollary}
    Given a step size sequence $\alpha_t\geq 0,t=0,1,2,\dots$ satisfying $\sum_{t=0}^\infty \alpha_t=\infty, \sum_{t=0}^\infty \alpha_t^2<\infty$, consider the update rule: $\mathbf{Q}_0=0; \mathbf{Q}_{t+1}=(1-\alpha_t)\mathbf{Q}_t+\alpha_t\sT\mathbf{Q}_t$,
    we have that $\mathbf{Q}_t$ converges to the unique fixed point $\mathbf{Q}^*$ of $\sT$ and the policy induced by the RQE of the stage games at each state $s$ converges to the RQE of the Markov game.
\end{corollary}
The proof of convergence follows from standard results in \cite{borkar2008stochastic}. Notice that given $\mathbf{Q}$ we can apply $\sT$ by first computing $\texttt{RQE}_i$ using methods in solving monotone games.

\section{CONCLUSIONS}
In this work we prove uniqueness and Lipschitz continuity of RQE in normal form games, and tractability of RQE in discounted infinite-horizon Markov games under monotonicity assumptions. While we focused on the two-player case, generalization to multi-player case should be straightforward. In the Markov game setting, it is also possible to develop provably convergent policy-gradient based algorithms by similarly defining policy evaluation Bellman operators. It is also interesting to see how our proposed algorithm behaves empirically in multi-agent reinforcement learning, even if the underlying game does not satisfy strong monotonicity.

\addtolength{\textheight}{-12cm}   




\bibliographystyle{ieeetr}
\bibliography{references}

\end{document}